\documentclass{article}

\usepackage{url}
\usepackage[english]{babel}
\usepackage[latin1]{inputenc}
\usepackage{color}
\usepackage{pgf}
\usepackage{graphicx}
\usepackage{times}
\usepackage{enumerate}
\usepackage{amsmath}
\usepackage{amsfonts}
\usepackage{amssymb}
\usepackage{amsthm}

\newtheorem{theorem}{Theorem}[section]
\newtheorem{corollary}[theorem]{Corollary}
\newtheorem{lemma}[theorem]{Lemma}

\newtheorem{definition}[theorem]{Definition}

\newtheorem{example}[theorem]{Example}

\newcommand{\Z}{{\mathbb{Z}}} \newcommand{\Q}{{\mathbb{Q}}} 
\newcommand{\C}{{\mathbb{C}}} \newcommand{\N}{{\mathbb{N}}}
 \newcommand{\K}{{\mathbb{K}}}

\newcommand{\ord}{\mathrm{ord}}
\newcommand{\gcrd}{\mathrm{GCRD}}
\newcommand{\Hom}{\mathrm{Hom}}
\newcommand{\cS}{\circledS}
\newcommand{\genexp}{\mathrm{GenExp}}
\newcommand{\valg}{\mathrm{ValG}}
\newcommand{\trunc}{\mathrm{Trunc}}
\newcommand{\ind}{\mathrm{Ind}}
\newcommand{\Seq}{\mathbf{S}}
\newcommand{\Gq}{\mathrm{Gquo}}
\newcommand{\SNF}{\mathrm{SNF}}

\begin{document}

\title{Closed form solutions of linear difference equations in terms of symmetric products}

\author{Yongjae Cha\thanks{Supported by the Austrian Science Fund (FWF) under grant P22748-N18}\\
Johannes Kepler University\\
      4040 Linz (Austria)\\
\texttt{ycha@risc.jku.at}}

\maketitle

\begin{abstract}
In this paper we show how to find a closed form solution for third
order difference operators in terms of solutions of second order
operators. This work is an extension of previous results on finding
closed form solutions of recurrence equations and a counterpart to
existing results on differential equations. As motivation and
application for this work, we discuss the problem of proving
positivity of sequences given merely in terms of their defining
recurrence relation. The main advantage of the present approach to
earlier methods attacking the same problem is that our algorithm
provides human-readable and verifiable, i.e., certified proofs.
\end{abstract}

\section{Introduction}
This paper presents an extension of the algorithm {\em solver}~\cite{YC11,CH09,CHG10} that
returns closed form solutions for second order linear difference equations to third order linear 
difference equations.  The solutions that we are looking for are in terms of (finite) sums
of squares. This is motivated by applying the algorithm for proving inequalities on
special functions, i.e., on expressions that may be defined in terms of linear difference
equations with polynomial coefficients. Conjectures about positivity of special functions
inequalities arise in many applications in mathematics and science. Proving them usually
requires profound knowledge on relations between these special functions. It is well known
that there exist many algorithms for proving and finding special function
identities~\cite{Zeil90a,ChyzakDM,AeqB,KoutschHF}. For automated proving of special
functions inequalities only few approaches exist. Gerhold and
Kauers~\cite{GKIneq,MKSumCracker} introduced a method that is based on Cylindrical
Algebraic Decomposition (CAD). This method has been proven to work well on many non-trivial
examples~\cite{MKTuran,VPSI}, but even though correctness is easy to be seen, termination
cannot be guaranteed, hence it is not an algorithm in the strict sense. A first attempt to
clarify the latter issue has been made in~\cite{MKVP10}. One of the features of proofs of
special functions identities is that they usually come with a certificate, i.e., some easy
to check identity that verifies the proof. The CAD-based approach can not hope to have a
similar certificate in the near future. The method presented here is a first step toward
human readable proofs of special functions inequalities, although admittedly a
representation in terms of sums of squares with positive coefficients is not expected to
exist for any given input. Besides this application, the results presented
  are of independent interest as they provide difference case counterparts to results obtained
  for the differential case~\cite{MS85, vH07}. 

First we review the available results in the differential case.  Let $k$ be a differential
field and $L_d \in k[\partial],\partial=d/dx$ be a linear homogeneous third order
differential operator.  Singer~\cite{MS85} characterizes when solutions of $L_d$ can be
written in terms of solutions of a second order operator in $\bar{k}[\partial]$.  Van
Hoeij~\cite{vH07} handles the similar problem when the coefficients of the second order
operator are restricted to $k$ and shows that it will be either of the following cases.
\begin{description}
\item[Case 1] $L_d$ is the symmetric square of a second order operator $K_d \in
  k[\partial]$
\item[Case 2] $L_d$ is gauge equivalent to a symmetric square of a second order operator
  $K_d \in k[\partial]$
\end{description} 
The definitions of symmetric products and gauge equivalence are recalled in sections~\ref{sec:sp}
and~\ref{sec:ge} below.  The algorithm given in~\cite{vH07} returns a second order differential
operator, $K_d \in k[d/dx]$, and a gauge transformation in $k[\partial]$ that sends
solutions of the symmetric squares of $K_d$ to solutions of $L_d$ for Case 2.

In the differential case, the symmetric square of $L_d$ has order~5 if and only if we are
in Case~1. In this case, there is a simple formula that gives $K_d$.  Case~2 is equivalent to the
symmetric square of $L_d$ having order~6 and a first order right-hand side factor
in~$k[\partial]$ as well as a certain conic of $L_d$(\cite[Equation 4.2.1]{MS85}) having a
non-zero solution in~$k$.  Since for $k=\C(x)$ this conic is solvable over $\C(x)$, the
last condition becomes trivial in this case.  The algorithm given in~\cite{vH07} in the
first step checks the order of the symmetric square of $L_d$ to distinguish the cases.

The difference case behaves differently; here we denote by $D=\C(x)[\tau]$ the ring of
linear difference operators, where $\tau$ denotes the shift operator.
Example~\ref{example:symp} shows that the cases can not be distinguished according to the
order of the symmetric squares when the coefficients are in~$\C(x)$. 
To set up a counterpart theorem for difference equations, 
this example shows that we need one more transformation than that in the differential case. 
Furthermore in Case~1, the algorithm for
finding the second order operator is more complicated than in the differential case.


Summarizing, the ideas used in the differential case can not be carried over immediately to
the difference case. Furthermore our aim is to have a closed form solution of the given
input. Hence, if a factorization is found that is not solvable, this fails to satisfy our
goal. Thus we build on the ideas of the algorithm {\em solver}~\cite{YC11,CH09,CHG10}.
Here we say that a function is in closed form if it is a linear combination of elementary
functions, special functions or hypergeometric functions over $\C(x)$.  For instance the
modified Bessel function of the first kind is a closed form solution of the second order
operator~$L_b:=z\tau^2-(2x+2)\tau+x+z$.

The algorithm {\em solver} returns closed form solutions for second order linear
difference operators. The main idea of {\em solver} is to map the given operator~$L_1$ to
an operator~$L_2$ of which a solution is known. This transformation is a bijective map,
called GT-transformation, that sends solutions of $L_1$ to solutions of $L_2$. If a closed
form solution to one of the operators is known, then by means of this transformation the
solution of the second operator can be constructed. For this purpose a table with
second order operators including parameters together with characteristic data (local data)
has been constructed. This local data can be computed for the given operator, the
corresponding equivalent operator is found by table look-up. Then by comparing parameters
of the local data the GT-transformation can be constructed. The characteristic data is
described in Section~\ref{sec:ld}. To cover the extension described here the table has been extended 
so that we can give closed form solutions of certain third order linear difference operators.

\section{Preliminary}
In this section we introduce notations used in this paper and recall some known
facts~\cite{CH09, CHG10, AeqB, PS97} about difference operators. Additionally Cases~1
and~2 above are carried over to the difference case for algebraic extensions in
Theorem~\ref{thm:ord3dec} below.

\subsection{Ring of difference oprators}
Let $D:=\C(x)[\tau]$ be the ring of linear difference operators with coefficients in $\C(x)$, where $\tau$ is the shift
operator acting on $x$ by $ \tau(x)=x+1$. Then $D$ is a noncommutative ring where $$\tau
\cdot \tau^{i-1}=\tau^i \ \text{for} \ i \in \N, \ \tau \cdot f= \tau(f)\tau \ \text{for}
\ f \in \C(x).$$ For $L=a_d(x)\tau^d+\cdots+a_1(x)\tau+a_0(x) \in D$ with $a_d \neq 0$, we
say that $L$ has order~$d$ and write $\ord(L)=d$. If furthermore $a_0\ne0$ then $L$ is
said to be a {\em normal} operator.  In this paper we will assume all operators to be
normal.

The adjoint operator of $L$ is defined by $L^* = \sum_{i=0}^d a_{d-i}(x+i)\tau^i$.
Suppose $L=M\cdot N$ for some $M, N \in D$.  Then $L^*=(M\cdot N)^*=(\tau^{d_1}\cdot
N^*\cdot \tau^{-d_1})\cdot M^*$, where $d_1=\ord(M)$ and thus right-hand side factors of $L$
correspond to left-hand side factors of~$L^*$. We say that a third order operator $L$ is
irreducible in $D$ if both $L$ and $L^*$ have no first order right-hand side factor in $D$.

A second order operator $K=b_2\tau^2+b_1\tau+b_0 \in D$ is called a {\em
    full} operator if $b_2b_1b_0 \neq 0 $.  Thus, if $K$ is a normal but not full
  operator, then $b_1=0$.

\subsection{Ring of sequences}

Let $\C^\N:=\{ f \mid f:\N \rightarrow \C\}$. Then an element $v \in \C^\N$
  corresponds to a sequence $v:=(v(1), v(2), v(3), \ldots)$. $\C$ is embedded in $\C^\N$
  as a subring via constant sequences.  Suppose $v_1, v_2 \in \C^\N$, then we define
\begin{equation*}
\begin{split}
v_1+v_2&:=(v_1(1)+v_2(1), v_1(2)+v_2(2), \ldots ) \\
 v_1v_2&:=(v_1(1)v_2(1), v_1(2)v_2(2), \ldots ).
\end{split}
\end{equation*}
With the above termwise addition and multiplication, $\C^\N$ forms a $\C$-algebra.  We
define the action of $\tau$ on $\C^\N$ by $\tau(v):=(v(2), v(3), v(4), \ldots)$.

Let $\mathbf{S} := \C^{\N}/_\sim$ where $s_1 \sim s_2$ if there exists $N \in \N$ such
that, for all $i > N$, $s_1(i) = s_2(i)$.  Then it is easy to verify that $s$ is a unit in
$\Seq$, i.e. $s$ is invertible in $\Seq$, if and only if $s \in \Seq$ has only finitely
many zeros.  If $f \in \C(x)$, then the image of $f$ in $\Seq$ and the action of $\tau$ on
$\Seq$ are well defined. This way we can embed $\C(x)$ to $\Seq$ and call $s \in \Seq$
rational if there exist $g(x) \in \C(x)$ and $N \in \N$ such that $g(i)=s(i)$ for all $i
\geq N$.  $\Seq[\tau]$ forms a ring of difference operators and $D$ is embedded in
$\Seq[\tau]$.

We say $L(v)=0$ for $v \in \Seq, L=a_d(x)\tau^d+\cdots+a_0(x) \in \Seq[\tau]$ if there is $n_0 \in \N$ such that
$$a_d(i)v(i+d)+a_{d-1}(i)v(i+d-1)+\cdots+a_0(i)v(i)=0 \quad \text{for \ all} \ i \geq n_0.$$

\begin{definition}
$h \in \Seq$ is called hypergeometric if $r=\tau(h)/h \in \Seq \setminus \{0\}$ is rational and 
$r$ is called the certificate of $h.$
\end{definition}
If $h \in \Seq$ is hypergeometric then $(\tau-r)(h)=0$ where $r$ is the certificate of~$h$.
We define $V(L):=\{ u\in \mathbf{S} \mid L(u)=0 \}$.

\begin{theorem}\cite[Theorem 8.2.1]{AeqB}
$\dim_\C(V(L))=\ord(L)$ for a normal difference
operator $L \in D$. 
\end{theorem}

Thus for a normal operator $L \in D$, $V(L)$ forms a $\C$-vector space 
with a basis $\{ v_i \in \Seq \mid 1 \leq i \leq \ord(L) \}$.


\subsection{Term equivalence}\label{sec:sp}

\begin{definition}\label{def:sp}
  The symmetric product, $M \circledS N$, of operators $M$ and $N \in D$ is an
  order-minimal and monic operator such that $\mu \nu \in V(M \circledS N)$ for all $\mu
  \in V(M)$ and $\nu \in V(N)$. 
\end{definition}

There is a simple formula if one of the operators has order~$1$.  Let $L
=a_d(x)\tau^d+\cdots+\dots+a_1(x)\tau+a_0(x) \in D$ and $r(x) \in \C(x)$. Then
\begin{equation}
\label{sym}
\begin{split}
& L \circledS (\tau-r(x))=\displaystyle \sum_{i=0}^d b_i \tau^i, \ \text{where} \ b_d(x)=a_d(x) \ \text{and}\\
& b_i(x)=a_i(x)\displaystyle\prod_{j=i}^{d-1} \tau^j(r(x)) \ \text{for} \ i=0,\ldots,d-1.
\end{split}
\end{equation}
Thus, $(\tau-a(x)) \circledS (\tau-b(x)) =\tau-a(x)b(x)$ for any $a(x), b(x) \in \C(x)$.

Suppose $L \in D$ and $s \in \Seq$. Then the above formula gives an operator $\tilde{L}=L \cS (\tau-s) \in \Seq[\tau]$ such that $V(\tilde{L})=\{ hu \mid L(u)=0 \}$ 
where $h \in \Seq$ is a solution of $\tau-s$. If $L \cS (\tau-s) \in D$ then it is easy to see that $s$ is rational.

\begin{definition}
  $L_1, L_2 \in D$ are said to be term equivalent if there exists $T=\tau-r \in D$ such that
  $V(L_2)=V( L_1 \cS (\tau-r))$, denoted by $L_1 \sim_t L_2$. Such a T is called the term
  transformation from $L_1$ to $L_2$.
\end{definition}

If $L_1$ and $L_2$ are term equivalent and $\tau-r $ is the term transformation then
$V(L_2)=\{ h v \mid h \in V(\tau-r), v \in V(L_1) \}$.  Suppose $L_1$ and a term
transformation $T$ are given, then $L_2$ can be obtained by~\eqref{sym}.

\subsection{Gauge equivalence}\label{sec:ge}
Let $L_1,L_2\in D$ be two given operators, where a closed form solution~$u$ of $L_1$ is
known. If furthermore an operator $G\in D$ can be determined sucht that $G(u)$ is solution
of~$L_2$, then a closed form solution of $L_2$ can be written as a linear combination of
shifts of $u$ over $\C(x)$. Such a transformation~$G$ is called a gauge transformation and
$L_1$ and $L_2$ are said to be gauge equivalent if such a transformation exists.


\begin{definition}\label{def:gt}
  Let $L_1, L_2 \in D$ have the same order. $G \in D$ is called a {\em gauge
    transformation} from $L_1$ to $L_2$ iff $G: V(L_1) \rightarrow V(L_2)$ is a bijection.
\end{definition}

Note that $G$ is not required to be a normal operator.  

Suppose we are given a gauge transformation $G$ where $\ord(G) \geq \ord(L_1)$. Then there exist $Q, \hat{G} \in D$
with $\ord(\hat{G}) < \ord(L_1)$ such that~$G=QL_1+\hat{G}$. The remainder $\hat{G}$ is
also a gauge transformation that acts in the same way as $G$ on $V(L_1)$. Hence, 
w.l.o.g., we may assume that~$\ord(G)<\ord(L_1)$.

Let $\gcrd(L, M)$ denote the greatest common right divisor of $L, M \in D$.  
Since $G$ is a bijection, any non zero solution $u$ of $L_1$ does not map to zero by $G$. 
Thus, $L_1$ and $G$ have no nontrivial common right hand factor, i.e. $\gcrd(L_1, G)=1$. 
Using the extended Euclidean algorithm $\tilde{G},\tilde{L_1}\in D$ can be determined such
that $\tilde{G}G+\tilde{L_1}L_1=1$. Then $\tilde{G}G$ is the identity on $V(L_1)$ and
$\tilde{G}$ is an inverse of $G$ that sends $V(L_2) \rightarrow V(L_1)$ bijectively.

\begin{definition}\label{def:ge}
  Two operators $L_1$ and $L_2$ with the same order are called gauge equivalent if there
  exists a gauge transformation $G:V(L_1)\rightarrow V(L_2)$ and we use the notation~$L_1
  \sim_g L_2$.
\end{definition}

Suppose $L_1 \sim_g L_2$ where the gauge transformation from $L_1$ to $L_2$ is a single term operator, $c(x)\tau^n$ for $n < \ord(L_1)$. 
Then $\tau^n\cdot L_1 \cdot\tau^{-n}$ is term equivalent to $L_2$ where the term transformation from $\tau^n\cdot L_1 \cdot \tau^{-n}$ to $L_2$ is $\tau-\frac{c(x+1)}{c(x)}$.


\subsubsection{How to compute the gauge transformation}\label{hom}
Suppose $L_1$ and $L_2$ are gauge equivalent and $G$ is a gauge transformation from $L_1$
to $L_2$.  Then there is an operator $H \in D, \ord(H) < \ord(L_2)$ such that $H \cdot
L_1=L_2 \cdot G$.

The algorithm that was used to find the gauge transformation in~\cite{CH09, CHG10, GH10}
works as follows:
\begin{enumerate}
\item For given operators $L_1$ and $L_2$, we set up the ansatz
  $G:=\sum_{i=0}^{\ord(L_1)-1} c_i(x)\tau^i$, where the $c_i(x)$ are undetermined
  coefficients.
\item\label{step3} Right divide $L_2 \cdot G$ by $L_1$ and set the remainder to zero. This way we obtain a
  system $A$ of difference equations for the unknown coefficients $c_i(x)$.
\item  Compute the  rational  solutions of the  system   $A$ to determine  the  values for
  the~$c_i(x)$.
\end{enumerate}

This algorithm was efficient for second order operators, but for operators
  of order three and higher, computing a solution of the system $A$, we get at
  Step~\ref{step3}, is very costly. Hence in the current implementation we use the new
algorithm HOM to compute the gauge transformations that give the set of homomorphisms
$\Hom_D(V(L_1),V(L_2))$ in~$D$ sending $V(L_1)$ to $V(L_2)$ for any~$L_1,L_2\in D$. This
means in particular that we can drop the condition on the orders,~$\ord(L_1) = \ord(L_2)$.

In short, the algorithm HOM works as follows: For $L=\sum_{i=0}^d a_i(x)\tau^i \in D$,
$a_d(x)=1$, we define the $\vee$-adjoint operator $L^\vee:=\sum_{i=0}^d
a_{d-i}(x+i-1)\tau^i$.  Then there is a one to one correspondence between $\Hom(L_1, L_2)$
and rational (invariant under the difference Galois group) elements of $V(L_1^\vee)
\otimes V(L_2)$.  We define a space $\mathcal{M}(L_1^\vee, L_2)$ that is isomorphic to
$V(L_1^\vee) \otimes V(L_2)$.  Then rational elements of $\mathcal{M}(L_1^\vee, L_2)$
correspond bijectively to elements of $\Hom(L_1, L_2)$. Thus, we compute rational elements
of $\mathcal{M}(L_1^\vee, L_2)$.  This is done by working directly with $L_1^\vee$ and
$L_2$, and we avoid computing large operators such as the symmetric product of $L_1^\vee$
and $L_2$ (whose solution space is a homomorphic image of $\mathcal{M}(L_1^\vee, L_2)$.)

Note that if $L_1$ and $L_2$ are of the same order, then HOM returns exactly the gauge
transformations. The algorithm HOM is available at
\url{http://www.risc.jku.at/people/ycha/Hom.txt} and more details can be found in \cite{TensorRatSol}. This is joint work of Yongjae Cha and Mark van Hoeij.

\subsection{GT-equivalence}

\begin{definition}
  Suppose there is a gauge transformation $G$ and a term transformation $T=\tau-r(x)$ such that the
  composition of $G$ and $T$, $G \circ T$, maps $V(L_1)$ to $V(L_2)$, i.e. 
   $G:  V(L_1 \cS (\tau-r(x))  \rightarrow V(L_2)$.  
  Then $L_1$ and $L_2$ are called GT-equivalent, denoted by~$L_1 \sim_{gt}
  L_2$, and the composition of $G$ and $T$ is refered to as the GT-transformation from
  $L_1$ to $L_2$.
\end{definition}

Suppose there is a map $\overline{GT}$ which is a multiple composition of gauge transformations and term transformations. Then 
\cite[Theorem 3.3.]{Le10} shows that we can find a gauge transformation $G$ and a term transformation $T$ such that
$\overline{GT}(V(L_1))=G\circ T( V(L_1))$.

\subsubsection{How to compute the GT-Transformation}

\begin{definition}\label{SNFdef}
  Let  $C$ be a subfield of $\C$ and $r(x) = cp_1(x)^{e_1} \dotsm p_j(x)^{e_j} \in C(x)$, for some $e_i \in \Z$,
  monic irreducible in $p_i(x) \in C[x]$, and let $s_i \in C$ equal
  the sum of the roots of $p_i(x)$.

  $r(x)$ is said to be in {\em shift normal form} if $-\deg(p_i(x)) < {\rm Re}(s_i)
  \leqslant 0$, for $i=1, \dotsc, j$. We denote by $\SNF(r(x))$ the shift normalized form
  of $r(x)$, which is obtained by replacing each $p_i(x)$ by $p_i(x+k_i)$ for some $k_i
  \in \Z$ such that $p_i(x+k_i)$ is in shift normal form.
\end{definition}

$\SNF(r(x))$ is unique up to the choice of $C$. In the algorithm given in
Section~\ref{sec:algo} we assume $C=\Q$.  For $L = a_d(x)\tau^d+ \cdots + a_0(x) \in D $,
we denote by $\det(L)$ the determinant of the companion matrix of $L$, which is
$(-1)^da_0(x)/a_d(x)$.

\begin{theorem}\cite[Theorem 2.3.9]{YC11}\label{tpsuff}
  Suppose $L_1 \sim_{gt} L_2$ for $L_1, L_2 \in C(x)[\tau]$ where $C$ is a subfield of
  $\C$.  Then there exists a gauge transformation $G \in C(x)[\tau]$ from $L_1 \otimes
  (\tau-r(x))$ to $L_2$ for some $r(x) \in C(x)$ where $$r(x)^d =
  \SNF(\det(L_2)/\det(L_1)), \quad \ord(L_1)=d.$$
\end{theorem}

The original statement of the above theorem uses $C=\C$, but the same proof works for any
subfield $C$ of $\C$.

Suppose we know that $L_1 \sim_{gt} L_2$ for $L_1, L_2 \in \Q(x)[\tau]$ and we want to
find the GT-transformation.  By the above theorem there exists $r(x) \in \Q(x)$ such that
$\SNF(\det(L_2)/\det(L_1))=r(x)^d$ where $d=\ord(L_1)$.  When $d$ is even $L_1 \cS
(\tau-r(x))$ or $L_1 \cS (\tau+r(x))$ can be gauge equivalent to $L_2$.  Thus the
algorithm Hom will return a non-empty set for either of the two. Furthermore, when $d$ is
odd, $L_1 \cS (\tau-r(x))$ is gauge equivalent to $L_2$.

\subsection{Symmetric powers of operators}
Given an operator $L\in D$ that annihilates a function $u$, then in order to obtain an
operator $M\in D$ that annihilates~$u^2$ we need the symmetric square of~$L$. In this
section we state some facts about these operators.

By $L^{\circledS m}$ we denote the $m^\text{th}$ symmetric power of $L$, i.e.,we
define~$L^{\circledS 1}=L$ and $L^{\circledS m} = L{\circledS}L^{\circledS(m-1)}$.  $K$ is
called a symmetric square root of $L$ if $L=K^{\circledS 2}$.

Suppose $K_d$ is a differential operator of order 2 then it is known that the order of
$L_d^{\circledS m}$ is $m+1$~\cite[Lemma 3.2, (b)]{MS85}.  However the following lemma
shows that this is not true for difference operators.

\begin{lemma}\cite[Lemma 3]{GH10}\label{lemma:sym-power2}
 Let $K= a_2(x)\tau^2+a_1(x)\tau+a_0(x) \in D$. Then
\begin{enumerate}
 \item 
	if $a_1(x) \neq 0$ then 
	\[
        K^{\circledS 2}=b_3(x)\tau^3+b_2(x)\tau^2+b_1(x)\tau+b_0(x),
        \] 
        where
        \[
        \begin{aligned}
	\null\kern-1.5em b_3(x)&=a_1(x)a_2(x+1)^2a_2(x)\\
	\null\kern-1.5em b_2(x)&=a_1(x+1)a_2(x)(a_0(x+1)a_2(x)-a_1(x+1)a_1(x))\\
	\null\kern-1.5em b_1(x)&=a_0(x+1)a_1(x)(a_1(x+1)a_1(x)-a_0(x+1)a_2(x))\\
	\null\kern-1.5em b_0(x)&=-a_1(x+1)a_0(x+1)a_0(x)^2.
        \end{aligned}
        \]
\item if $a_1(x) =0$ then $K^{\circledS 2}=a_2(x)^2\tau^2-a_0(x)^2$.
\end{enumerate}
The formulas above give order-minimal operators for both cases.
\end{lemma}

If a full operator $K=a_2(x)\tau^2+a_1(x)\tau+a_0(x)$ is
a symmetric square root of a third order operator $L$, then also $\overline{K} =
a_2(x)\tau^2-a_1(x)\tau+a_0(x)$ is a symmetric square root of~$L$. If $u$ is a solution
of $K$, then $(-1)^x u$ is a solution of~$\overline{K}$. We say $K$ and $\overline{K}$ are conjugates if $K \sim_t \overline{K}$ where the term transformation is $\tau+1$.

Solutions of an equation of type 2 are called Liouvillian solutions~\cite{CH09,HS99,GH10}.
Suppose $u_1$ is a solution of $K=a_2(x)\tau^2-a_0(x)$ then $\{ u_1, u_2 \}$, where $u_2=
(-1)^xu_1$, forms a basis of $V(K)$ and~$u_1^2=u_2^2$.  Also, it is easy to verify that
for arbitrary orders~$m$ it holds that $K^{\circledS m}=a_2(x)^m\tau^2+(-1)^{m+1}a_0(x)^m$
with a similar proof to the one of Lemma~\ref{lemma:sym-power2}.

\begin{definition}
  A second order operator $K$ is called a unity free operator if the solution space of
  $K$ does not admit a basis $\{ v_1, v_2 \}$ such that $v_1^n=v_2^n$ for some $n \in \N$.
\end{definition}

Let $K=\left( {x}^{2}+x \right) {\tau}^{2}+ \left( 2\,x+{x}^{2} \right)
\tau+{x}^{2}+3\,x+2$.  Then a basis of the solution space of $K$ is $\{xw_1^x, xw_2^x \}$
where $w_1$ and $w_2$ are solutions of $z^2+z+1$ in $\C$.  Since $ (xw_1^x)^3=
(xw_2^x)^3=x^3$ for $x \in \N$, $K$ is not a unity free operator.

\begin{lemma}
\label{lm:unfr}
If $K \in D$ is an irreducible second order operator then it is a unity free operator.
\end{lemma}

\begin{proof}

  We prove this by contraposition.  Suppose $K \in D$ is not a unity free operator.
  Then we may assume $K$ is monic and $V(K)$ admits a basis $\{ v_1, v_2 \}$ such that
  $v_1^n=v_2^n$ for some $n \in \Z$.  Let $n_0 \in \Z_{>0}$ be the smallest integer that
  satisfies $v_1^{n_0}=v_2^{n_0}$, then we may assume $v_1=(u_{n_0}^a)^xf,
  v_2=(u_{n_0}^b)^xf$ for some $f \in \Seq$ where $u_{n_0}$ denotes $n_0$th root of unity
  and $n_0, a, b$ are pairwise relatively prime.  Thus
  $K=(\tau^2-(u_{n_0}^a+u_{n_0}^b)\tau+u_{n_0}^au_{n_0}^b) \cS (\tau-r)$ where
  $r=\tau(f)/f$. Since $K$ is an element in $D$ and by equation~\eqref{sym}, $r$ is in
  $\C(x)$ and this implies $K$ is reducible in D.
\end{proof}

\begin{lemma}
\label{lm:zd}
If $v \in \Seq, v \neq 0$ satisfies a full second order operator $K=b_2(x)\tau^2+b_1(x)\tau+b_0(x) \in D$ then 
$v$ is not a zero divisor in $\Seq$.
\end{lemma}

\begin{proof}
We will prove that $v$ has only finitely many zeros.
Since $K(v)=0$ there is $n_0 \in \N$ such that 
\begin{equation}\label{eq:zd}
b_2(x)v(x+2)+b_1(x)v(x+1)+b_0(x)v(x)=0
\end{equation}
and $b_i(x)$ has no poles or roots for all $x \geq n_0,i=1,2$. Suppose $v(n_1)=v(n_1+1)=0$
for some $n_1 \geq n_0$. Then by~\eqref{eq:zd}, $v(x)=0$ for all $x \geq n_1$ and this
contradicts that $v \neq 0$. Suppose $v(n_2)=0, v(n_2+1) \neq 0$ for some $n_2 \geq n_0$.
Then again by~\eqref{eq:zd}, $v(x) \neq 0$ for all $x \geq n_2+1$. Thus $v(x) \neq 0$ for
$x$ large enough and hence $v$ is a unit.
\end{proof}

\begin{theorem}\label{thm:nonvan}

  If $L=K^{\cS m}$ for some irreducible full second order operator $K \in D$ then $\ord(L)=m+1$

\end{theorem}
\begin{proof} 

   Let $\{ v_1, v_2 \}$ be a basis of $V(K)$. We will show that then $\{v_1^iv_2^{m-i} \mid
  i=0..m \}$ are linearly independent.  Suppose there exist $c_i$ in $\C$, not all zero
  such that $c_m v_1^m+c_{m-1} v_1^{m-1}v_2+\cdots+c_0 v_2^m=0$.  
  
  By Lemma~\ref{lm:zd}, $v_2$ is not a unit and since $K$ is irreducible operator, by
  Lemma~\ref{lm:unfr}, $v_1^{n}/v_2^{n} \neq 1$ for any $n \in \N$.  Let $z:=v_1/v_2 \in
  \Seq$ and $f(y):=c_m y^m+c_{m-1}y^{m-1}+\cdots+c_0$ then $f(z)=0$, i.e. $f(z(x))=0$ for
  all $x \in \N$.  Thus, $z(x) \in \{ c \in \C \mid f(c)=0 \}$ for all $x \in \N$ and
  $v_1=zv_2$. Suppose $z$ is not a constant sequence. Since $K$ is an irreducible full operator
  in $D$, it contradicts that $v_1$ is a solution of $K$.  Suppose $z$ is a constant
  sequence. Then it contradicts that $v_1$ are $v_2$ linearly independent.
\end{proof}

In the differential case it holds that if the symmetric square of a third order
differential operator $L_d \in \C(x)[\partial]$ has order 5, then $L_d=K_d^{\cS 2}$ for
some second order operator $K_d \in \C(x)[\partial]$.  However, the following example
shows that this does not hold in the difference case.

\begin{example}
\label{example:symp}
Let $E:=( x+1) {\tau}^{3}+ ( -28{x}^{3}-4{x}^{4}-36-84 x-73{x}^{2}) {\tau}^{2}+ (
-69x-77{x}^{3}-18-104{x }^{2}-4{x}^{5}-28{x}^{4} ) \tau+{x}^{4}+5{x}^{3}+8{x}^{2 }+4x \in
D$. Then $\ord(E^{\cS 2})=5$. A solution of $E$ is $x I_x(1)^2$ where $I_x(z)$ denotes the
modified Bessel function of the first kind.  Then the symmetric square roots of $E$ are
$K_{1}=\tau^2+(2+2x)\sqrt{x+1}\tau-\sqrt{x(x+1)}$ and
$K_{2}=\tau^2-(2+2x)\sqrt{x+1}\tau-\sqrt{x(x+1)}$, which are not in~$D$.  A solution of
$K_{1}$ and $K_{2}$ are $\sqrt{x}I_x(1)$ and $-\sqrt{x}I_x(1)$, respectively.

Let $B:=z\tau^2-(2x+2)\tau+x+z$. Then a solution of $B$ is $I_x(z)$.  Then $E=K_{1}^{\cS2}
\sim_t B^{\cS2}$, but $B$ and $K_{1}$ are not gauge equivalent in~$D$, i.e, there is no
operator in $D$ that sends $V(B)$ to $V(K_{1})$.
Since $\sqrt{x}$ is not a solution of 
any shift operator in $D$, \cite[Theorem 5.2]{FlajoletGerholdSalvy05}and \cite[Lemma A.2]{Chen2012111}, $K_{1}$ is not a symmetric product of $B$ and a difference
operator in~$D$.

\end{example}

In the differential case, suppose $L_d \sim_t K_d^{\cS 2}$, i.e, the
  solution of $L_d$ can be obtained by multiplying a hyperexponential term $h$ to the
  solutions of $K_d^{\cS 2}$. Let $\{ u_1, u_2\}$ be a basis of the solution space of
  $K_d$, then $L_d$ admits a basis of the solution space $\{ gu_1^2, gu_1u_2, gu_2^2\}$.
  However, if $g$ is hyperexponential then $\sqrt{g}$ is also hyperexponential. Thus,
  $L_d=\tilde{K}^{\cS 2}$ for $\tilde{K} \in \C(x)[\delta]$ such that $\tilde{K}=K \cS
  (\partial-\frac12 \frac{g'}{g})$ where $g'=\frac{d}{dx} h$.  However if $h$ is a
  hypergeometric term, $\sqrt{h}$ is not guaranteed to be a solution of an operator in $D$.


\begin{definition}\label{def:gauge-equ}
  An irreducible operator $L$ is said to be {\emph solvable in terms of second order in $D$} if it is
  GT-equivalent to $K_1 \circledS K_2 \cdots \circledS K_d$ where the $K_i$'s
  are irreducible and full second order operators in $D$.
\end{definition}

We need the following Lemma to prove Theorem~\ref{thm:ord3dec}.

\begin{lemma}\label{symbs}
  Let $K_1, K_2 \in D$ be full second order operators. 
 If $\ord(K_1 \cS K_2)=3$ then 
 we can choose a basis $\{ v_1, v_2 \}$ of $V(K_1)$, and a basis $\{
  w_1, w_2 \}$ of $V(K_2)$, such that $v_1w_2=v_2 w_1$.

\end{lemma}

\begin{proof}
  Let $\{ v_1, v_2 \}$ be a basis of $V(L_1)$ and $\{ w_1, w_2 \}$ be a basis of $V(L_2)$.
  Since $\ord(K_1 \cS K_2)=3$, the $\C$-vector space generated by $\{ v_1w_1, v_1w_2,
  v_2w_1, v_2w_2 \}$ has dimension 3. Then there exists $a_1, a_2, a_3 \in \C$, which are
  not all zero, such that $$v_1w_2=a_1v_1w_1+a_2v_2w_1+a_3v_2w_2.$$ Suppose $a_1=a_2=0$
  and $a_3 \neq 0$ then it contradicts that $v_1$ and $v_2$ are linearly independent.
  Likewise, if $a_2= a_3=0$ and $a_1 \neq 0$ then it contradicts that $w_1$ and $w_2$ are
  linearly independent.  If $a_1=a_3=0$ and $a_2 \neq 0$ then we have the desired form.
  So, the remaining cases are either only one of the coefficients $a_1, a_2, a_3$ is zero,
  or all $a_1, a_2, a_3$ are non-zero.  Here, we will prove the case when $a_2$ is the
  only zero coefficient.  Let $\{\tilde{w}_1, \tilde{w}_2 \}$ be another basis of $V(L_2)$
  such that
$$
\begin{pmatrix}
\tilde{w}_1 \\
\tilde{w}_2
\end{pmatrix}
=
\begin{pmatrix}
0 & a_3 \\
-a_1 & 1
\end{pmatrix}
\begin{pmatrix}
w_1 \\
w_2
\end{pmatrix}
$$
Then for $\{\tilde{w_1},\tilde{w_2}\}$ we have $v_1\tilde{w_2}=v_2\tilde{w_1}$ as claimed.
\end{proof}

\begin{theorem}\label{thm:ord3dec}
  Let $L$ be an operator of order 3, irreducible and solvable in terms of second order in~$D$.
  Then $L \sim_{gt} K^{\circledS 2}$ for some irreducible full second order operator $K \in D$
  and furthermore

\begin{enumerate}[(a)]

\item\label{c1}  $L \sim_t K^{\cS 2}$ then $\ord(L^{\circledS 2})=5$.

\item\label{c2} if the gauge transformation of $L \sim_{gt} K^{\cS 2}$ is not a
  single term operator then $\ord(L^{\circledS 2})=6$.
\end{enumerate}

\end{theorem}

\begin{proof} Let $L$ be a third order, irreducible operator that is
    solvable in terms of second order in~$D$.  Then by definition (and the restriction of
    the order), there exist two irreducible full second order operators $K_1, K_2 \in D$
    such that $L\sim_{gt}K_1\cS K_2$. By Lemma~\ref{symbs}, a suitable basis $\{v_1,v_2\}$
    for $K_1$, and a suitable basis $\{w_1,w_2\}$ for $K_2$ can be chosen, such
    that~$v_1w_2=v_2w_1$.  Let $h=w_1/v_1=w_2/v_2$, then $h \in \Seq$ and $w_1=hv_1$, and
    $w_2=hv_2$.  Since $\{ w_1, w_2\}$ is a basis for an operator in $D$, $h$ is
    hypergeometric and this implies that $K_1 \sim_t K_2$ with term transformation
    $\tau-r$, where $r$ is the certificate of the hypergeometric term~$h$. Summarizing, by
    Lemma~\ref{lemma:sym-power2}, $L \sim_{gt} K^{\circledS 2}$ for some full operator $K
    \in D$.

  (a) Let $\{ v_1, v_2 \}$ be a basis of $V(K)$ and $\tau-r$ be the term transformation from $K^{\cS 2}$ to $L$ and $h$ be a solution of $\tau-r$. Then
  $\{ hv_1^2, hv_1v_2, hv_2^2 \}$ forms a basis of $L$ and thus $\ord(L^{\circledS 2})=5$ by Lemma~\ref{thm:nonvan}.

  
  (b) Let $\{ v_1, v_2 \}$ be a basis of $V(K)$. Then $\{ G(hv_1^2), G(hv_1v_2), G(hv_2^2)
  \}$ forms a basis of $V(L)$, where $G=c_2(x)\tau^2+c_1(x)\tau+c_0(x) \in D$ is a non
  single term operator and $h$ is a hypergeometric term.  Then $G(hv_1^2)G(hv_2^2) \neq
  G(hv_1v_2)^2$ and this implies $\ord(L^{\circledS 2})=6$.
 \end{proof}

 Suppose $L_d$ is a differential operator of order 3 and $L_d = K_d^{\cS 2}$ for some
 second order differential operator $K_d$. Then it is well known that there exists a
 formula to construct this~$K_d$, see~\cite[Lemma 3.4]{MS85}. The case where only
 gauge-equivalence holds, i.e., $L_d\sim_g K_d^{\cS2}$, is more interesting.
 In~\cite{vH07} third order operators are treated with a focus on determining both $K_d$
 and a gauge transformation.

 It is possible to implement a similar algorithm for difference equations which returns
 the second order operator $K$ to which the given $L$ can be reduced to and a gauge
 transformation.  However, in the difference case, in order to give a closed form solution
 of $K$ other algorithms need to be applied or a table look-up. Also, even if we are in
 case~\eqref{c1}, finding $K$ is not as simple as in the differential case, in particular
 if there is a parameter included in the input.  Morever to distinguish the cases, the
 symmetric square of a third order operator needs to be computed which can become costly
 if many parameters are involved.



\section{Local data}\label{sec:ld}

The local data that we are using are the valuation growths at finite singularities in
$\C/\Z$ and generalized exponents at the point of infinity.  This data is invariant under
GT transformations.  In this section, we review the definition and an invariance
property (Theorem~\ref{gp}, Theorem~\ref{genexp}) of local data from~\cite{YC11, CH09,
  CHG10, HO99}. We omit proofs in this paper.

\subsection{Finite singularities}
Valuation growth was first introduced in~\cite{HO99} and an algorithm to compute it was
given in the same paper. Let $L=a_d\tau^d+\cdots+a_0\tau^0 \in D$. After multiplying
$L$ from the left by a suitable element of $\C(x)$, we may assume that the
 $a_i$ are in $\C[x]$ and gcd$(a_0,\ldots,a_d)=1$. Then $q \in \C$ is called a {\em
   problem point} of $L$ if $q$ is a root of the polynomial $a_0(x)a_d(x-d)$ and $p \in
 \C/ \Z$ is called a {\em finite singularity} of $L$ if $L$ has a problem point in $p$
 (i.e. $p=q+\Z$ for some problem point $q$).  Let $p \in \C/\Z$.  For $a,b \in p \subset
 \C$ we say $a > b$ iff $a-b$ is a positive integer.

 Let $\varepsilon$ be a new indeterminant, i.e., transcendental over $\C$.  We define
 $L_\varepsilon:=\sum_{i=0}^d a_i(x+\varepsilon)\tau^i$ which is obtained by substituting
 $x$ with $x+\varepsilon$ in $L$. The map $L \mapsto L_\varepsilon$ defines an embedding
 (as non-commutative rings) of $\C(x)[\tau]$ in $\C(x,\varepsilon)[\tau]$. Hence, if
 $L=MN$, then~$L_\varepsilon=M_\varepsilon N_\varepsilon$.


\begin{definition} Let $a \in \overline{C}(\epsilon)$ and $\overline{C}[[\epsilon]]$ be the ring of 
formal power series over $\overline{C}$ in $\epsilon$. The $\varepsilon$-valuation
  $v_\varepsilon(a)$ of $a$ at $\varepsilon=0$ is the element of $\Z \cup {\infty}$
  defined as follows: if $a\neq0$ then $v_\varepsilon(a)$ is the largest integer $m \in
  \Z$ such that $a/\varepsilon^m \in \overline{C}[[\varepsilon]]$, and
  $v_\varepsilon(0)=\infty$. \end{definition}

We define an $\ord(L)$ dimensional $\C(\varepsilon)$-vector space 
$$V_p(L_\varepsilon):=\{\tilde{u}:p \rightarrow \C(\varepsilon) \mid L_\varepsilon(\tilde{u})=0\}.$$
Let $q_l$ be the smallest root of $a_0(x)a_d(x-d)$ in $p$, so
$q_l$ is the smallest problem point
in $p$. Likewise we define $q_r$ to be the largest root of $a_0(x)a_d(x-d)$ in $p$. If $p$ is not a singularity,
that is, if $a_0$ and $a_d$ have no roots in $p$,
then choose two arbitrary elements in $p$ and define $q_l, q_r$ to be those two elements.

\begin{definition} \label{box} For non-zero $\tilde{u} \in V_p(L_\varepsilon)$ and for $a,
  b \in \C$ if $b=a+d-1$, where $d=\ord(L_\varepsilon)$, we define the {\em box-valuation}
$$v^a_b(\tilde{u})=\min\{v_\varepsilon(\tilde{u}(m))|m=a,a+1,\ldots,b\}.$$ \end{definition}

\begin{lemma}
    \label{vl}
With $q_l, q_r$ chosen as above, we have
$$v_{q-1}^{q-d}(\tilde{u})=v_{q_l-1}^{q_l-d}(\tilde{u}) \ \ \text{\rm for all} \ q \in \{q_l-1,q_l-2,q_l-3,\ldots \},$$
$$v_{q+d}^{q+1}(\tilde{u})=v_{q_r+d}^{q_r+1}(\tilde{u}) \ \ \text{\rm for all} \ q \in \{q_r+1,q_r+2,q_r+3,\ldots \}$$. \end{lemma}

We define $v_{\varepsilon,l}(\tilde{u})$ as $v_{q_l-1}^{q_l-d}(\tilde{u})$ which, by Lemma~\ref{vl}, equals the box valuation of any box on the left of
$q_l$. Likewise we define $v_{\varepsilon,r}(\tilde{u})$ as $v_{q_r+d}^{q_r+1}(\tilde{u})$.

\begin{definition}
\label{def:valg}
 Define the {\em valuation growth} of non-zero $\tilde{u} \in V_p(L_\varepsilon)$ as
$$g_{p,\varepsilon}(\tilde{u})=v_{\varepsilon,r}(\tilde{u})-v_{\varepsilon,l}(\tilde{u}) \in \Z.$$
Define the {\em set of valuation growths} of $L$ at $p$ as
$$\overline{g}_p(L)=\{g_{p,\varepsilon}(\tilde{u}) \mid \tilde{u}\in V_p(L_\varepsilon),\tilde{u}\neq0\}\subset \Z.$$ \end{definition}
If $L$ is a first operator operator then $\overline{g}_p(L)$ has only one element.

\begin{definition}
    \label{apartsing}
    Let $L$ be a difference operator and $p \in \C/\Z$ be a finite singularity of $L$. If
    $\overline{g}_p(L)$ has more than one element then $p$ is called an {\em essential
      singularity}. 
\end{definition}

The algorithm given in~\cite{HO99} determines the set 
\[
\{ \overline{g_p}(L) \mid p \text{ \  is an essential singularity of}\  L \}.
\] 

\begin{theorem}
    \label{gp}\cite[Theorem 1]{CH09}
If $L_1$ and $L_2$ are gauge equivalent then 
$$\max(\overline{g_p}(L_1))=\max( \overline{g_p}(L_2) ) \quad \text{and} \quad \min(\overline{g_p}(L_1))=\min(\overline{g_p}(L_2))$$
for all $p \in \C/\Z$.  
\end{theorem}


The following lemma is an immediate consequence of Definition~\ref{def:valg}.
\begin{lemma}
For each $p \in \C/\Z$, 
$$\max(\overline{g_p}(L^{\cS 2})=2\max(\overline{g_p}(L)) \quad \text{and} \quad
\min(\overline{g_p}(L^{\cS 2})=2\min(\overline{g_p}(L)).$$
\end{lemma}

The above theorem only gives invariance under gauge equivalence. 
To have invariance under GT-equivalence, we need to define one more set. 
Suppose $L_1 \sim_{gt} L_2$, then $L_1 \cS (\tau-r(x)) \sim_g L_2$ for some $r(x) \in \C(x)$. 
Then $$\max(\overline{g_p}(L_2))=\max( \overline{g_p}(L_1) )+d \quad \text{and} \quad \min(\overline{g_p}(L_2))=\min( \overline{g_p}(L_1) )+d$$
where $\{ d\}=\overline{g_p}(\tau-r(x))$, $d \in \Z$. So $d_p(L)=\max(\overline{g_p}(L))-\min(\overline{g_p}(L))$ is invariant under GT-equivalence.
Thus, for a difference operator $L \in $, we define a set of ordered pairs
$$\valg:=\{ (p, d_p(L)) \in \C/\Z \ \times \  \Z_{\geq 0} \mid p \text{ \ is an essential singularity of \ } L \}.$$

\subsection{Singularity at infinity}

Let $\K:=\C((t)), x=1/t$ be the field of formal Laurent series and $\K_r=\C((t^{1/r}))$ for $r \in \N$. We define the valuation
for $a \in \K$ as the smallest power of $a$ whose coefficient is non-zero and denote it
by~$v(a)$.  This definition can be extended to $\hat{D}=\K[\tau]=\K[\Delta]$, where
$\Delta:=\tau-1$ denotes the forward difference, by setting
\[
v(L)=\min \{ v(a_i)+i\mid L=a_0+\cdots+a_d \Delta^d\}
\]
for any operator~$L\in\hat{D}$.

\begin{lemma}
\label{ind}
Let $L \in \K[\tau]$.  There exists a polynomial $P$ such that
for every $n \in \Z$ we have
\begin{equation}\label{tn}
	L(t^n) = P(n) t^{n + v(L)} + \cdots
\end{equation}
where the dots refer to terms of valuation $> n + v(L)$.
\end{lemma}

\begin{definition}
	$\ind_L(n)$, the {\em indicial polynomial} of $L$, is the polynomial $P(n)$ in Lemma~\ref{ind}~\eqref{tn}.
\end{definition}

Lemma 9.2 in \cite{CHG10} states that if $N \in \Z$ is a root of $\ind_L(n)$ then there is $u \in
\K$ such that $L(u)=0$ and $v(u)=N$.  However, there is no one-to-one correspondence
between solutions of $L$ in $\K$ and integer roots of
$\ind_L(n)$. 
For this matter, we introduce the ring $\K[l]$, where $l$ is a solution of $\tau(l) - l =
t$, see~\cite{LF99} for existence of $l$.  We extend valuation on $\K$ to $\K[l]$ by: for
$a=a_1t^d+\cdots \in \K[l]$, $a_i \in \C[l]$, $d \in \Z$, and $a_1 \neq 0$, we
let~$v(a)=d$.  With this notion we obtain the following theorem which is equivalent
to~\cite[Theorem 3.2.10]{YC11} and \cite[Lemma 6.1]{PS97}.

\begin{theorem}
 $p \in \Z$ is a solution of $\ind_L(n)$ if and only if $L$ has a solution $u \in \K[l]$ with $v(u)=p$. 
\end{theorem}

An immediate consequence of the above theorem is the following corollary.

\begin{corollary}
\label{indsym}
If $p_1$ and $p_2 \in \Z$ are the solutions of the indicial equations of $L_1$ and $L_2$,
respectively, then $p_1+p_2$ is a solution of the indicial equation of~$L_1 \circledS L_2$.
\end{corollary}


Define the action of $\tau$ on $\K_r$ as:
\begin{equation}
\label{acttau}
\begin{split}
\tau(t^{\frac1r}) &= t^{\frac1r} (1+t)^{-\frac1r}\\
            &=t^{\frac1r} (1 - \frac{1}{1!}\frac1r t + \frac{1}{2!}\frac1r(\frac1r+1) t^2 \\
            & -\frac{1}{3!}\frac1r(\frac1r+1)(\frac1r+2) t^3 + \cdots ) \in \K_r.
\end{split}
\end{equation}

Since we have defined the action of $\tau$ on $\K_r$, we can now apply the formula for the
term symmetric product in~\eqref{sym} to $\K_r[\tau]$.  Let $E_r$ and $\tilde{G}_r$ be the
following subset and subgroup, respectively, of $\K_r^*$:
$$
E_r=\biggr\{a \in K_r^* \mid a=ct^v(1+\displaystyle\sum_{i=1}^r a_i t^{i/r}), a_i \in \C,
c \in \C^*, v \in \tfrac1r\Z\biggl\},
$$
$$\tilde{G}_r=\biggr\{ a \in K_r^* \mid a=1+\displaystyle\sum_{i=r+1}^{\infty} a_i t^{i/r} , \ a_i \in \C \biggl\}.$$

Now $E_r$ is a set of representatives for $\K^*_r / \tilde{G}_r$. The composition of the
natural maps $\K^*_r \rightarrow \K^*_r/\tilde{G}_r \rightarrow E_r$ defines a natural map
$$ \trunc:\K^*_r \rightarrow E_r .$$
Let
$$G_r=\{ a \in \K_r^* \mid a=1+\frac{m}{r}t+\displaystyle\sum_{i=r+1}^{\infty} a_i t^{i/r} , \ a_i \in \C, \ m \in \Z \}.$$

\begin{definition}\label{def:req}
Let $r \in \N$ then for $a, b \in E_r$, we say $a$ is $r$-equivalent to $b$, $a \thicksim_r b$, when $a/b \in G_r$.
\end{definition}

Note that $a \thicksim_r b$ if and only if $a_r \equiv b_r\!\mod \frac1r \Z$ with $a_r$ as
in the definition of $E_r$, $a_i=b_i$ for $i < r$, and $c,v$ matching as well.

\begin{definition}
  Let $g \in E_r$ for some $r \in \N$. We say that $g$ is a {\em generalized exponent} of
  $L$ with multiplicity $m$ if and only if zero is a root of $\ind_{\tilde{L}}(n)$ with
  multiplicity m where $\tilde{L}= L \circledS (\tau-\frac1g)$. We denote by $\genexp(L)$
  the set of generalized exponents of~$L$.
\end{definition}

Suppose $L=\tau-r(x) \in D$ then $\genexp(L)=\{ \trunc(r(t)) \}$.


\begin{theorem}
\label{genexp}
If two operators $L_1$ and $L_2$ are gauge equivalent then for each $g_1 \in
\genexp(L_1)$ there is a $g_2 \in \genexp(L_2)$ such that $g_2$ is equivalent to $g_1$.
\end{theorem}
This theorem has been proven first in~\cite{CHG10}. An alternative proof can be found in~\cite{YC11}.


\begin{theorem}\label{symgen}
Suppose $L, L' \in D$ then 
$$
\genexp(L \circledS  L')= \{ \trunc(gg') \mid  g \in  \genexp(L), 
g' \in \genexp(L')   \}
$$
\end{theorem}

\begin{proof}
  $L \circledS L' \circledS (\tau-\frac{1}{gg'})=L \circledS (\tau-\frac{1}{g}) \circledS
  L' \circledS (\tau-\frac{1}{g'})$ and since 0 is a solution of $L \circledS
  (\tau-\frac{1}{g})$ and $L \circledS (\tau-\frac{1}{g'})$, 0 is also a solution of
  the indicial equation of $L \circledS L' \circledS (\tau-\frac{1}{\trunc(gg')})$ by
  Lemma~\ref{indsym}
\end{proof}

Likewise for the valuation growth, we need to define one more set to have invariance for GT-equivalence.
Suppose $L_1 \cS (\tau-r(x)) \sim_g L_2$ for some $r(x) \in \C(x)$. Then 
$$\genexp(L_2)=\{ g_r g \mid g \in \genexp(L_1), \ \{g_r\}=\genexp(\tau-r(x)) \}.$$

Thus we define the following set,
$$\Gq(L):=\{ \trunc(g_i/g_j) \mid g_i \neq g_j,  g_i, g_j \in \genexp(L)\}$$
and then $\Gq(L_1)=\Gq(L_2)$ if $L_1 \sim_{gt} L_2$.

\section{Table of base equations}


In \cite{YC11, CHG10}, we have formed a table of base equations of order 2, call it TB, as follows;
\begin{itemize}
\item collect equations with known solution from \cite{abst, AG76}.
\item for any closed form expression that shows up often in the literature, generate a base equation with existing algorithms \cite{ChyzakDM, KoutschHF}.
\end{itemize}
For the algorithm given in Section~\ref{sec:algo}, we have computed symmetric squares of each base equation in TB yielding an entry in TB2 of a base equations of third orders.  
Moreover we have generated further base equations as follows:

Suppose $u(x)$ is a solution of an operator $L=\sum_{i=0}^d a_i(x)\tau^i$.
Then $u(x/m)$ is a solution of the operator
\begin{equation} \label{eq:Tm}
 L_{(m)}=\sum_{i=0}^d a_i(x/m)\tau^{mi}.
\end{equation}
As input for our algorithm we accept only operators of order three and the above equation
may be of higher order. One way of obtaining the base equation for $u(x/m)$ in this case
is using $L_{(m)}$ when it is a multiple of an operator $M\in D$ for which~$M(u(x/m))=0$.
Since $L_{(m)}$ as constructed above is not guaranteed to be the minimial order operator
we compute~$\Hom(L_{(m)},L_{(m)})$. If the algorithm HOM returns the identity map this means
that $L_{(m)}$ is in fact order-minimal. These cases are neglected and
we use $L_{(m)}$ as a base equation only if HOM returns a non-trivial map.

For instance for the squared hypergeometric function in the table below,
${}_2F_1\Big[\genfrac{}{}{0pt}{}{-x/2+a, \ x/2+b}{c} ; z \Big]^2$, an annihilating
operator $L_{(2)}$ can be obtained starting from an operator $L_{(1)}$ annihilating
${}_2F_1\Big[\genfrac{}{}{0pt}{}{-x+a, \ x+b}{c} ; z \Big]^2$ using~\eqref{eq:Tm}.  Then
the order-minimality of $L_{(2)}$ is checked with the algorithm HOM. In this case HOM returns
a non-identity map and hence we save $L_{(2)}$ in the table.


If $ct^vf \in \genexp(L)$, then $z_c\trunc(g_m^vf) \in \genexp(L_{(m)})$, where $z_c$ is a root of
$z^m=c$ and $g_m^v \in \genexp(\tau^m-(\frac{x}{m})^v)$.  Thus, we can detect whether an
input operator may have a solution $u(x/m)$ if a base equation for $u(x)$ is in our table.
However, it is more efficient to compute the base equation for small values of~$m$.


\subsection{Example of base equations}
\label{sec:base}

Here we list a small part of the table which is needed in Section~\ref{sec:algo} and \ref{apps}.  In the
following table they are listed under (a) a solution (b) the corresponding $\Gq$, and (c) the $\valg$. 
The full table can be found at
\url{http://www.risc.jku.at/people/ycha/TB2.txt}.

\begin{enumerate}

\item\label{ex2} 	
	\begin{enumerate}
	  \item ${}_2F_1\Big[\genfrac{}{}{0pt}{}{-x/2+a, \ x/2+b}{c}  ; z \Big]^2$
	  \item $ \left\{ -1,- \left( 2\,z-1\pm2\,\sqrt {{z}^{2}-z} \right) ^{2}, \pm(2\,z-1\pm2\,\sqrt {{z}^{2}-z}) \right\} 
 $
	  \item $\{ (-2b, 2), (2a, 2), (2a-2c, 2), (2c-2b, 2) \}$
	\end{enumerate}

\item 
	\begin{enumerate}\label{Legendre}
	  \item $P_x(z)^2$ (Legendre polynomials squared)
	  \item $\left\{ -1+2\,{z}^{2}\pm2\,\sqrt {-{z}^
{2}+{z}^{4}}, \left( -1+2\,{z}^{2}\pm2\,\sqrt {-{z}^{2}+{z}^{4}} \right) ^{
-1}, \frac{ -1+2\,{z}^{2}\mp2\,\sqrt {-{z}^{2}+{z}^{4}}}{ -1
+2\,{z}^{2}\pm2\,\sqrt {-{z}^{2}+{z}^{4}} } \right\} 
$
	  \item $\{(0, 4) \}$
	\end{enumerate}

\item
	\begin{enumerate}
	  \item $H_x(z)^2$ (Hermite polynomials squared)
	  \item $\left\{ -1\pm\sqrt {-2\,{z}^{2}}T+{z}^{2}{T}^{2},1\pm2\,\sqrt {-2\,{z}^{2}}T-4\,{z}^{2}{T}^{2}\right\} 
$
	  \item $\{(0, 2) \}$
	\end{enumerate}
\end{enumerate}

\section{Algorithm}\label{sec:algo}

The basic structure of the algorithm is the same that was given in \cite{YC11}. Here we use an extended table of base equations and a more efficient algorithm for computing
the gauge transformation, as mentioned in Section~\ref{hom}.

Suppose $L$ is the input operator with local data
$$ \Gq(L)=\{ a, \overline{a}, b, \overline{b}, c, \overline{c} \} , \quad \valg(L)=\{ (0,4)  \}$$
for $a, b, c \in \C$.
By comparing the corresponding data in TB2, we can find that local data of $L$ matches with the data of \eqref{Legendre} in Section~\ref{sec:base}.
Let $L_{lgd}$ be the operator of which $P_x(z)^2$ is a solution. To compute the parameter $z$, we compare $a$ with each entry of $\Gq(L_{lgd})$ 
and compute the set of candidates of possible values for $z$ which is, 
$$
\left\{ \pm\frac12\,\sqrt {{\frac {2\,a\pm\sqrt {2\,{a}^{2}+{a}^{3}+a}}{a}}}, \
\pm\frac12\,{\frac {a+1}{\sqrt {a}}} \right\}.
$$
Substituting $z$ by each of the values of the above set, a set of equations $cdd2$ is obtained. 
It remains to cheek for each of the equations in $cdd2$ whether there is a GT-transformation to $L$ and 
if so then we return the closed form solution by applying the GT-transformation to 
$P_x(z)^2$. 
\\

\noindent {\bf Algorithm {\em solver2}}\\
\noindent {\bf Input}: A third order normal operator $L_I \in \Q[x, \tau]$.\\
\noindent {\bf Output}: Either at least one closed form solution of $L$ in the form of $c_0(x)u(x)^2+c_1(x)u(x+1)^2+c_2(x)u(x+2)^2$ where 
$c_i(x)$ are hypergeometric terms and $u(x)^2$ is a solution in TB2 or otherwise the empty set.

\begin{enumerate}
\item $cdd1:=\{\}, GQ:=\Gq(L_I), VG:=\valg(L_I)$ .
\item Find the base equations in TB2 by comparing $GQ$ and $VG$ with the corresponding data in the table.
\begin{enumerate}
\item if there is no match then return `Not solvable within the Table'.
\item if there is a matching equation $L_c$, $cdd1:=cdd1 \cup \{L_c\}$.
\end{enumerate}
\item\label{four} For each $L_c \in cdd1$, compute candidate values for the parameters using $GQ$ and $VG$.
\item Construct a set $cdd2$ by substituting parameters by the values determined in Step.\ref{four}
\item For each $L_p \in cdd2$ check if there exists a GT-Transformation from $L_p$ to $L_I$.
\begin{enumerate}
\item if there is a GT-transformation then apply $GT$ to the known solution of $L_c$ and return the solution.
\item if there is no GT-transformation found return `Not solvable within the Table'. 
\end{enumerate}

\end{enumerate}


\section{Improvement}
A similar approach can be applied to higher order operators that are solvable in terms
of order two. Suppose $L_4$ is a fourth order operator that is solvable in terms of order 
two in~$D$. Then $L_4$ is equal or gauge equivalent to either $K_1^{\cS 3}$ or $K_1\cS K_2$
for some second order operators $K_1,K_2\in D$ with nonvanishing coefficients. The candidates
for $K_i$ can be detected analogously using Theorem~\ref{symgen}.

Concerning the applications to proving positivity of special functions inequalities it has
to be noted that representations in terms of finite linear combination of squares with non-negative
coefficients need not exist on the full range of validity of a given inequality, as can be
seen below. Further investigations of the applicability of this approach as well as an
implementation of the above mentioned extension to higher order recurrences are ongoing
work.

\section{Applications}
\label{apps}
Our main motivation to extend finding closed form solutions of difference equations in
terms of symmetric products is to develop an algorithmic approach for proving special
functions inequalities. Existing symbolic methods~\cite{GKIneq,MKSumCracker,MKVP10} are
based on using Cylindrical Algebraic Decomposition (CAD) which in several examples has
proven to be an effective way for proving positivity of sequences that are given only in
terms of their defining sequences. However, it is sometimes unsatisfiable to have a proof
that only comes with ``True'' without any certificate. Some classical proofs of
inequalities are using rewriting of the given expression as linear combination of easy to
verify positive objects such as sums of squares. The present work tries to make this
approach algorithmic. Certainly it will not provide answers for any special functions
inequality, but it is a first step in a new direction of automatic inequality proving.
Below we give two examples, one for each of the cases distinguished above, of classical
problems that can be solved fully or at least partially using the presented algorithm.
Note that all of these identities stated can be proven easily using existing algorithms
for symbolic summation. The novelty is the automatic discovery of certain closed form
expressions for sequences that are given only in terms of their defining recurrence
relation. In this sense it is comparable to the above mentioned algorithms based on CAD.

\subsection{Clausen's formula} \label{Clausen} 

Proofs of special function inequalities often depend on a variety of classical techniques
such as argument transformations, integral representations of hypergeometric series and
many more. For instance in the proof of the Askey-Gasper inequality~\cite{AG76},which
played a key role in the proof of the Bieberbach conjecture by
de~Branges~\cite{deBranges}, Clausen's formula
\begin{align}\label{cl}
&{}_3F_2\biggl[\genfrac{}{}{0pt}{}{-x, x+\alpha+1, \frac{\alpha+1}{2}}{\alpha+1, \frac{\alpha+2}{2}}  ; z \biggl] \\
&\quad= {}_2F_1\biggl[\genfrac{}{}{0pt}{}{-\frac12x, \frac12x+\frac12(\alpha+1)}{\frac{\alpha}{2}+1} ; z \biggl]^2\nonumber
\end{align}
entered at a central point. Zeilberger~\cite{SBE93} has shown how this identity can be
proven using symbolic summation. By means of the algorithm presented here, Clausen's formula
can be discovered entirely automatic.

The hypergeometric function in~\eqref{cl} satisfies a third order recurrence that is given 
by the operator~$L_3$ and is too large to be displayed here. It can however be found easily
common symbolic summation algorithms~\cite{Zeil90a, ChyzakDM, KoutschHF}. This difference operator
is the input for our procedure and we start by determining the local data given by
\[
\begin{aligned}
&\Gq(L_3)=  \left\{  -\left( 2\,z-1\pm2\,\sqrt {{z}^{2}-z} \right) ^{2},-2\,z+1\pm2\,
\sqrt {{z}^{2}-z} \right\}
,\\
& \valg(L_3)=\left\{ (0, 2),(-\alpha, 2) \right\}.
\end{aligned}
\]
A table look-up shows that this local data is compatible with~\ref{ex2} in Section~\ref{sec:base}. Comparing local data and solving the system mod $\Z$ 
the following candidates for $a,b$ and $c$ can be found:
\[
a\in\{0,\tfrac12\},\quad b \in \{\tfrac12 \alpha, \tfrac12 \alpha+\tfrac12\},\quad c \in \{\tfrac12 \alpha+1, \tfrac12\alpha+\tfrac32 \}.
\]
There is no term transformation for these operators and an application of HOM shows that we obtain a constant map if $a=0$, $b=\frac12
\alpha+\frac12$ and~$c=\frac12 \alpha+1$.

\noindent 

\subsection{Tur\'{a}n inequality for Hermite polynomials}
\label{turan}
The positivity of Tur\'{a}n determinants has been proven for many different families of
orthogonal polynomials. The first Tur\'{a}n inequality was formulated for Legendre
polynomials~\cite{TP50} and Szeg\"o~\cite{SG48} has given four different proofs of this
inequality.  Szwarc~\cite{Szwarc} has provided a more general approach for proving
Tur\'{a}n type inequalities based on the mere knowledge of the recurrence coefficients
satisfied by the given sequence. Gerhold and Kauers~\cite{MKTuran} have proven and
improved this type of inequalities using their CAD-based method. The approach presented
here does not give a full proof for Tur\'{a}n type inequalities, however it gives a
representation of the given determinant in sums of squares derived from the third order
annihilating operator of the determinant. In the case of Hermite polynomials this yields a
representation that gives positivity in the limit for $n$ tending to infinity.

Tur\'{a}n's inequality for Hermite polynomials $H_x(z)$ reads as follows:
\[
\Delta_x(z)=H_{x+1}(z)^2- H_{x}(z)H_{x+2}(z) \ge 0,\quad n\ge0,\ z\in\mathbb{R}.
\]
Then an annihilating operator of $\Delta_x(z)$
is $L_h:=\tau^3+(2x+2-4z^2)\tau^2-4(x+2)(x-2z^2+4)\tau-8(1+x)(x+2)^2$ and the local
data of this operator is
\[
\begin{aligned}
&\Gq(L_h)=  \left\{ -1\pm\sqrt {-2\,{z}^{2}}t^\frac12+ \left( {z}^{2}\pm1 \right) t,
1\pm2\,\sqrt {-2\,{z}^{2}}t^\frac12-4\,{z}^{2}t \right\} 
,\\
& \valg(L_h)=\left\{ (0, 2) \right\}.
\end{aligned}
\]

$-1\pm\sqrt {-2\,{z}^{2}}t^\frac12+ \left( {z}^{2}\pm1 \right)t$ are elements in $\Gq(L_h)$ and
these are equivalent to $-1\pm\sqrt {-2\,{z}^{2}}t^\frac12+  {z}^{2}t$ under $\sim_2$, see Definition~\ref{def:req} for $\sim_2$. 
Thus the local data of $L_h$ correspond to those of the third entry of the table given in Section~\ref{sec:base}.

Using the algorithm described above a gauge transformation can be found that applied to 
$H_x(z)^2$ yields the following equivalent formulation
\[
\Delta_x(z) = \tfrac12 H_{x+1}^2(z)+2(x+1-z^2) H_x(z)^2 + 2x^2 H_{x-1}^2(x).
\]
This representation gives the positivity of Tur\'{a}n's inequality for
$$z\in[-\sqrt{x+1},\sqrt{x+1}\ ],\quad x\ge0.$$

\bibliographystyle{plain}
\bibliography{myrefs}

\begin{thebibliography}{10}

\bibitem{abst}
M.~Abramowitz and I.~A. Stegun.
\newblock {\em Handbook of Mathematical Functions with Formulas, Graphs, and
  Mathematical Tables}.
\newblock Dover, New York, ninth dover printing, tenth gpo printing edition,
  1964.

\bibitem{AG76}
R.~Askey and G.~Gasper.
\newblock {Positive Jacobi Polynomial Sums, II}.
\newblock {\em American Journal of Mathematics}, 98(3):pp. 709--737, 1976.

\bibitem{YC11}
Y.~Cha.
\newblock {\em Closed Form Solutions of Linear Difference Equations}.
\newblock PhD thesis, Florida State University, Tallahassee, FL, USA, 2011.

\bibitem{CH09}
Y.~Cha and M.~van Hoeij.
\newblock {Liouvillian Solutions of Irreducible Linear Difference Equations}.
\newblock In {\em ISSAC '09: Proceedings of the 2009 international symposium on
  Symbolic and algebraic computation}, pages 87--94, New York, NY, USA, 2009.
  ACM.

\bibitem{TensorRatSol}
Y.~Cha and M.~van Hoeij.
\newblock Rational elements of the tensor product of solutions of difference
  operators.
\newblock In {\em Proceedings of the tenth Asian Symposium on Computer
  Mathematics}, 2012.

\bibitem{CHG10}
Y.~Cha, M.~van Hoeij, and G.~Levy.
\newblock Solving recurrence relations using local invariants.
\newblock In {\em ISSAC '10: Proceedings of the 2010 International Symposium on
  Symbolic and Algebraic Computation}, pages 303--309, New York, NY, USA, 2010.
  ACM.

\bibitem{Chen2012111}
S.~Chen and M.~F. Singer.
\newblock Residues and telescopers for bivariate rational functions.
\newblock {\em Advances in Applied Mathematics}, 49(2):111 -- 133, 2012.

\bibitem{ChyzakDM}
F.~Chyzak.
\newblock An extension of {Z}eilberger's fast algorithm to general holonomic
  functions.
\newblock {\em Discrete Math.}, 217(1-3):115--134, 2000.
\newblock Formal power series and algebraic combinatorics (Vienna, 1997).

\bibitem{deBranges}
L.~de~Branges.
\newblock A proof of the {B}ieberbach conjecture.
\newblock {\em Acta Math.}, 154(1-2):137--152, 1985.

\bibitem{SBE93}
S.B. Ekhad.
\newblock {A Short, Elementary and Easy WZ Proof of the Askey-Gasper Inequality
  That was Used by de Branges in his Proof of the Bieberbach Conjecture}.
\newblock {\em Theor. Comput. Sci.}, 117(1{\&}2):199--202, 1993.

\bibitem{FlajoletGerholdSalvy05}
P.~Flajolet, S.~Gerhold, and B.~Salvy.
\newblock On the non-holonomic character of logarithms, powers and the $n$th
  prime function.
\newblock {\em Electronic Journal of Combinatorics}, 11(2):1--16, 2005.

\bibitem{GKIneq}
S.~Gerhold and M.~Kauers.
\newblock {A Procedure for Proving Special Function Inequalities Involving a
  Discrete Parameter}.
\newblock In {\em {Proceedings of ISSAC '05}}, pages 156--162. ACM Press, 2005.

\bibitem{MKTuran}
S.~Gerhold and M.~Kauers.
\newblock {A Computer Proof of Turan's Inequality}.
\newblock {\em Journal of Inequalities in Pure and Applied Mathematics},
  7(2):1--4, May 2006.
\newblock Article 42.

\bibitem{HS99}
P.A. Hendricks and M.F. Singer.
\newblock Solving difference equations in finite terms.
\newblock {\em J. Symb. Comput.}, 27(3):239--259, 1999.

\bibitem{MKSumCracker}
M.~Kauers.
\newblock {SumCracker -- A Package for Manipulating Symbolic Sums and Related
  Objects}.
\newblock {\em Journal of Symbolic Computation}, 41(9):1039--1057, 2006.

\bibitem{MKVP10}
M.~Kauers and V.~Pillwein.
\newblock {When can we detect that a P-finite sequence is positive?}
\newblock In Stephen Watt, editor, {\em {Proceedings of ISSAC'10}}, pages
  195--202, 2010.

\bibitem{KoutschHF}
C.~Koutschan.
\newblock {HolonomicFunctions (User's Guide)}.
\newblock Technical Report 10-01, RISC Report Series, University of Linz,
  Austria, January 2010.

\bibitem{LF99}
A.H.M. Levelt and A.~Fahim.
\newblock Characteristic classes for difference operators.
\newblock {\em Compositio Mathematica}, 117(02):223--241, 1999.

\bibitem{Le10}
G.~Levy.
\newblock {\em Solutions of second order recurrence relations}.
\newblock PhD thesis, Florida State University, 2010.

\bibitem{AeqB}
M~Petkov\v{s}ek, H.S. Wilf, and D.~Zeilberger.
\newblock {\em A=B}.
\newblock AK Peters, Ltd., 1996.

\bibitem{VPSI}
V.~Pillwein.
\newblock {Positivity of certain sums over Jacobi kernel polynomials}.
\newblock {\em Adv. in Appl. Math.}, 41(3):365--377, 2008.

\bibitem{MS85}
M.F. Singer.
\newblock {Solving Homogeneous Linear Differential Equations in Terms of Second
  Order Linear Differential Equations}.
\newblock {\em American Journal of Mathematics}, 107(3):pp. 663--696, 1985.

\bibitem{SG48}
G.~Szeg{\"o}.
\newblock On an inequality of {P}. {T}ur\'an concerning {L}egendre polynomials.
\newblock {\em Bull. Amer. Math. Soc.}, 54:401--405, 1948.

\bibitem{Szwarc}
R.~Szwarc.
\newblock Positivity of {T}ur\'an determinants for orthogonal polynomials.
\newblock In {\em Harmonic analysis and hypergroups ({D}elhi, 1995)}, Trends
  Math., pages 165--182. Birkh\"auser Boston, Boston, MA, 1998.

\bibitem{TP50}
P.~Tur{\'a}n.
\newblock On the zeros of the polynomials of {L}egendre.
\newblock {\em \v Casopis P\v est. Mat. Fys.}, 75:113--122, 1950.

\bibitem{PS97}
M.~van~der Put and M.~F. Singer.
\newblock {\em Galois Theory of Difference Equations}, volume 1666.
\newblock Springer-Verlag, 1997.

\bibitem{HO99}
M.~van Hoeij.
\newblock Finite singularities and hypergeometric solutions of linear
  recurrence equations.
\newblock {\em J. Pure Appl. Algebra}, 139:109--131, 1998.

\bibitem{vH07}
M.~van Hoeij.
\newblock Solving third order linear differential equations in terms of second
  order equations.
\newblock In {\em Proceedings of the 2007 international symposium on Symbolic
  and algebraic computation}, ISSAC '07, pages 355--360, New York, NY, USA,
  2007. ACM.

\bibitem{GH10}
M.~van Hoeij and G.~Levy.
\newblock Liouvillian solutions of irreducible second order linear difference
  equations.
\newblock In {\em Proceedings of the 2010 International Symposium on Symbolic
  and Algebraic Computation}, ISSAC '10, pages 297--301, New York, NY, USA,
  2010. ACM.

\bibitem{Zeil90a}
D.~Zeilberger.
\newblock A holonomic systems approach to special functions identities.
\newblock {\em J. Comput. Appl. Math.}, 32(3):321--368, 1990.

\end{thebibliography}

\end{document}